\documentclass[a4paper,USenglish,cleveref,autoref,dvipsnames]{lipics-v2021}

\usepackage{preamble} %

\title{Hitting Geodesic Intervals in Structurally Restricted Graphs}

\author{Tatsuya Gima}
{Hokkaido University, Sapporo, Japan}
{gima@ist.hokudai.ac.jp}
{https://orcid.org/0000-0003-2815-5699}
{JSPS KAKENHI Grant Numbers
JP24K23847, 
JP25K03077. 
}

\author{Yasuaki Kobayashi}
{Hokkaido University, Sapporo, Japan}
{koba@ist.hokudai.ac.jp}
{https://orcid.org/0000-0003-3244-6915}
{JSPS KAKENHI Grant Numbers 
JP23K28034, 
JP24H00686, 
JP24H00697.
}

\author{Yuto Okada}
{Nagoya University, Nagoya, Japan
\and \url{https://yutookada.com/en/}}
{pv.20h.3324@s.thers.ac.jp}
{https://orcid.org/0000-0002-1156-0383}
{Supported by JST SPRING, Grant Number JPMJSP2125.}

\author{Yota Otachi}
{Nagoya University, Nagoya, Japan
\and \url{https://www.math.mi.i.nagoya-u.ac.jp/~otachi/}}
{otachi@nagoya-u.jp}
{https://orcid.org/0000-0002-0087-853X}
{JSPS KAKENHI Grant Numbers
JP22H00513, 
JP24H00697, 
JP25K03076, 
JP25K03077. 
}

\author{Hayato Takaike}
{Nagoya University, Nagoya, Japan}
{takaike.hayato@nagoya-u.jp}
{}
{}

\authorrunning{T. Gima,
Y. Kobayashi,
Y. Okada,
Y. Otachi, and
H. Takaike}

\Copyright{Tatsuya Gima,
Yasuaki Kobayashi,
Yuto Okada,
Yota Otachi, and
Hayato Takaike}

\ccsdesc[500]{Mathematics of computing~Graph algorithms}
\ccsdesc[500]{Theory of computation~Parameterized complexity and exact algorithms}

\keywords{Terminal monitoring set, Structural graph parameter, Geodesic interval}

\nolinenumbers 

\hideLIPIcs  

\EventEditors{John Q. Open and Joan R. Access}
\EventNoEds{2}
\EventLongTitle{42nd Conference on Very Important Topics (CVIT 2016)}
\EventShortTitle{CVIT 2016}
\EventAcronym{CVIT}
\EventYear{2016}
\EventDate{December 24--27, 2016}
\EventLocation{Little Whinging, United Kingdom}
\EventLogo{}
\SeriesVolume{42}
\ArticleNo{23}

\begin{document}
\maketitle


\begin{abstract}
Given a graph $G = (V,E)$, a set $T$ of vertex pairs, and an integer $k$,
{\ProblemName} asks whether there is a set $S \subseteq V$ of size at most $k$ such that for each terminal pair $\{u,v\} \in T$, 
the set $S$ intersects at least one shortest $u$--$v$ path.
Aravind and Saxena [WALCOM~2024] introduced this problem and showed several parameterized complexity results.
In this paper, we extend the known results in both negative and positive directions
and present sharp complexity contrasts with respect to structural graph parameters.

We first show that the problem is NP-complete even on graphs with highly restricted shortest-path structures.
More precisely, we show the NP-completeness on graphs obtained by adding a single vertex to a disjoint union of $5$-vertex paths.
By modifying the proof of this result, we also show the NP-completeness 
on graphs obtained from a path by adding one vertex and 
on graphs obtained from a disjoint union of triangles by adding one universal vertex.
Furthermore, we show the NP-completeness on graphs of bandwidth~$4$ and maximum degree~$5$ by replacing the universal vertex in the last case with a long path.
Under standard complexity assumptions, these negative results rule out fixed-parameter algorithms for most of the structural parameters studied in the literature
(if the solution size~$k$ is not part of the parameter).

We next present fixed-parameter algorithms parameterized by~$k$ plus modular-width and by~$k$ plus vertex integrity.
The algorithm for the latter case does indeed solve a more general setting that includes 
the parameterization by the minimum vertex multiway-cut size of the terminal vertices.
We show that this is tight in the sense that the problem parameterized by the minimum vertex multicut size of the terminal pairs is W[2]-complete.
We then modify the proof of this intractability result and show that the problem is W[2]-complete parameterized by~$k$ even in the setting where $T = \binom{Q}{2}$ for some $Q \subseteq V$.
\end{abstract}



\section{Introduction}

For vertices $u$ and $v$ in a graph $G$, the \emph{geodesic interval} $I_{G}[u,v]$ is 
the set of vertices that belong to at least one shortest $u$--$v$ path in $G$. 
In this paper, we study the problem {\ProblemName} that asks to find a minimum hitting set for a given family of geodesic intervals,
which is formalized as follows.
\begin{tcolorbox}
\begin{description}
  \setlength{\itemsep}{0pt}
  \item[Problem:] {\ProblemName} (\ProblemNameShort)
  \item[Input:] A graph $G = (V,E)$, a set $T$ of vertex pairs, and an integer $k$.
  \item[Question:] Is there $S \subseteq V$ with $|S| \le k$ such that $S \cap I_{G}[u,v] \ne \emptyset$ for each $\{u,v\} \in T$?
\end{description}
\end{tcolorbox}
\noindent
We also study the same problem on edge-weighted graphs, which we call {\WProblemName} (\WProblemNameShort).
We call a vertex belonging to $V(T) = \bigcup_{\{u,v\} \in T}\{u,v\}$ a \emph{terminal vertex}. 
We say that a vertex $w$ \emph{interrupts} a pair $\{u,v\}$ if $w \in I_{G}[u,v]$.
Also, we say that a vertex set $S$ \emph{interrupts} $\{u,v\}$ if one of the members of $S$ interrupts $\{u,v\}$.

The setting of {\ProblemNameShort} may have many practical applications.
For example, it can be seen as a facility location problem in a network $G$
with a set $T$ of departure-destination point pairs of transportation demands.
The goal here is to locate as small number of facilities (e.g., gas stations) as possible in such a way that,
for each demand $\{u,v\} \in T$, one can take a shortest $u$--$v$ path in $G$ that visits at least one facility.

{\ProblemNameShort} has been introduced recently by Aravind and Saxena~\cite{AravindS24},
who motivated the setting as a communication network monitoring problem (and called it \textsc{Terminal Monitoring Set}).
They initiated the study of the parameterized complexity of {\ProblemNameShort} and showed several positive and negative results (see \cref{sec:known-results}).
Our contribution in this paper is to considerably tighten the known tractability/intractability border of {\ProblemNameShort} with respect to structural graph parameters,
as we describe in \cref{sec:our-results}.


\subsection{Known results}
\label{sec:known-results}

\subparagraph{Observations.}
To illustrate connections to other problems,
we first observe that some known facts imply (in)tractability of some cases of {\ProblemNameShort}.
(Some of them will be useful later.)
\begin{enumerate}
  \item {\WProblemNameShort} parameterized by $|T|$ is fixed-parameter tractable.

  \item {\WProblemNameShort} on trees is polynomial-time solvable.

  \item If \textsc{Vertex Cover} is $\mathrm{NP}$-complete on a graph class $\mathcal{C}$, then so is {\ProblemNameShort}\@.

  \item {\ProblemNameShort} on complete graphs is $\mathrm{NP}$-complete.

  \item {\ProblemNameShort} on $r \times r$ grids is $\mathrm{NP}$-complete.

\end{enumerate}

The first one is an easy observation that this is a special case of 
\textsc{Hitting Set} parameterized by the number of subsets to be hit, which is fixed-parameter tractable
(see \cite[Theorem~6.1]{CyganFKLMPPS15}).
The second follows from the fact that the following more general problem is polynomial-time solvable (see, e.g., the discussion by Jansen~\cite{Jansen17}):
given a set of subtrees of a tree, find a minimum vertex set that intersects all subtrees in the set.

The third and fourth observations come from the equivalence of 
an instance $\langle (V,E), k \rangle$ of \textsc{Vertex Cover} and
the instances $\langle (V, E), \binom{V}{2}, k \rangle$ and $\langle (V, \binom{V}{2}), E, k \rangle$ of {\ProblemNameShort}, 
where $\binom{V}{2} = \{\{u,v\} \mid u, v \in V, u \ne v\}$.
The equivalence holds since the only way to interrupts an adjacent pair is to pick one of them.

The fifth fact can be shown by observing an interesting connection to a geometric problem.
\textsc{Piercing Rectangles} asks, given axis-parallel rectangles $R_{1}, \dots, R_{r}$ in the plane and an integer $k$,
whether there is a set of $k$ points that intersects all rectangles $R_{i}$ (see \cite{AgarwalHRS24}).
It is known that even if all coordinates are integers, each rectangle is a $2 \times 2$ square,
and all rectangles are placed in a region of $cr \times cr$ for some constant $c$,
this problem is NP-complete~\cite{FowlerPT81}.\footnote{%
In~\cite{FowlerPT81}, the problem is described as a point-set covering problem by $2 \times 2$ squares, 
which is equivalent to the one described here because in each problem one object hits or covers another object
if the $L_{\infty}$ distance between their centers is at most~$1$.}
Such an instance of \textsc{Piercing Rectangles} can be transformed to an equivalent instance $\langle G, T, k\rangle$ of {\ProblemNameShort} 
by setting $G$ to the $cr \times cr$ grid and adding the pair of two opposite corners of each rectangle $R_{i}$ to $T$.
Observe that a geodesic interval $I_{G}[u,v]$ for $u$ and $v$ is the set of vertices in the axis-parallel rectangle having $u$ and $v$ at opposite corners.

\subparagraph{Results of Aravind and Saxena~\cite{AravindS24}.}
Their main results can be summarized as follows.
\begin{enumerate}
  \item {\ProblemNameShort} is W[2]-complete parameterized by the solution size~$k$.

  \item {\ProblemNameShort} is fixed-parameter tractable parameterized 
  by $k$ plus cluster vertex deletion number.

  \item {\ProblemNameShort} is fixed-parameter tractable parameterized 
  by $k$ plus neighborhood diversity.

  \item {\WProblemNameShort} is fixed-parameter tractable parameterized by feedback edge set number.

  \item {\WProblemNameShort} is fixed-parameter tractable parameterized by vertex cover number.
\end{enumerate}

Their first result complements the fixed-parameter tractability with respect to $|T|$ as $k < |T|$ in nontrivial instances.
Note that this result implies the W[2]-completeness of {\WProblemNameShort} parameterized by $k$ even on complete graphs
(consider the reduction from {\ProblemNameShort} that replaces each nonedge with a huge-weight edge).
In their second and third results, having~$k$ in the parameter is necessary because of the NP-completeness of {\ProblemNameShort} on complete graphs.
Their fourth result generalizes polynomial-time solvability on trees.
Note that they did not explicitly mention the weighted case for feedback edge set number, but their algorithm works without any modifications for the weighted case as well.

Since vertex cover number is an upper bound on the solution size $k$,
the parameterization solely by vertex cover number is equivalent to the one by $k$ plus vertex cover number.
Thus, most of the known positive results (except for the one parameterized by feedback edge set number) have ``plus $k$'' in their parameters.

As open questions, Aravind and Saxena~\cite{AravindS24} asked for the complexity of {\ProblemNameShort} parameterized by
treewidth, pathwidth, feedback vertex set number, and distance to path forest.


\subsection{Our results}
\label{sec:our-results}

\subparagraph{NP-completeness of {\ProblemNameShort} on highly restricted graphs.}
Our first results demonstrate that {\ProblemNameShort} is NP-complete even if the target graph classes are highly restricted,
and thus answer the open questions of Aravind and Saxena~\cite{AravindS24} in a negative way.
More precisely, we show that the problem is NP-complete on the following classes:
\begin{enumerate}
  \item graphs of vertex-deletion distance~$1$ to the class of disjoint unions of $5$-vertex paths;
  \item graphs of vertex-deletion distance~$1$ to the class of paths;
  \item graphs of vertex-deletion distance~$1$ to the class of disjoint unions of triangles;
  \item graphs of bandwidth~$4$ and maximum degree~$5$.
\end{enumerate}

Our main technical contribution here is the first one,
which already implies that {\ProblemNameShort} is paraNP-complete for most of the structural parameters studied in the literature
(see \cref{fig:graph-parameters}~(left))
as the graphs there have distance to path forest~$1$ and vertex integrity~$6$.
The second one is an easy corollary of the first.
The third one, which improves the vertex integrity bound to~$4$, is based on the same idea as the first.
The fourth one can be obtained from the third by a minor modification.

\subparagraph{FPT with $k$ plus structural parameters.}
To cope with the rather hopeless situation discussed above, we consider the setting where the solution size~$k$ is also part of the parameter.
In this setting, we present the following positive results.
\begin{enumerate}
  \item {\ProblemNameShort} is fixed-parameter tractable parameterized by~$k$ plus modular-width.
  \item {\WProblemNameShort} is fixed-parameter tractable parameterized by~$k$ plus vertex integrity.
  \item {\WProblemNameShort} is fixed-parameter tractable parameterized by the minimum vertex multiway-cut size of
  the set of terminal vertices $V(T)$.
\end{enumerate}

The first result for $k$ plus modular-width generalizes the one for $k$ plus neighborhood diversity~\cite{AravindS24}.
The second and third results are achieved by a single algorithm that solves {\WProblemNameShort} 
when a small separator of a special kind is given as part of the input.
After presenting the algorithm, we show that the two cases can be handled by this algorithm
without the assumption that the separator is given as part of the input (see \cref{sssec:app-main-alg} for details).
Observe that the second result generalizes 
the one parameterized solely by vertex cover number~\cite{AravindS24} since vertex cover number is an upper bound of both $k$ and vertex integrity.
The third result generalizes the one parameterized by $|T|$,
which can be seen as an example of the conceptual framework of Jansen and Swennenhuis~\cite{JansenS24}
for parameterizing terminal-involved problems with a parameter structurally smaller than the number of terminals.
Although the third result does not explicitly have~$k$ in the parameter,
we can see that its parameter is an upper bound of~$k$.

\cref{fig:graph-parameters} (right) shows
the complexity of {\ProblemNameShort} parameterized by solution size~$k$ plus a structural parameter.
As we can see, it is not well understood yet, and many cases remain unsettled.

\subparagraph{Parameterized intractability beyond $k$.}
As mentioned above, there are many unsettled cases in the ``plus~$k$'' setting.
Actually, we do not know any intractability result like ``{\ProblemNameShort} is W-hard parameterized by $k$ plus X'' for any structural graph parameter X\@.

As intermediate steps toward this direction, we present two intractability results for the following cases, where 
$k$ is part of the parameter, implicitly or explicitly.
\begin{enumerate}
  \item {\ProblemNameShort} is W[2]-complete parameterized by the minimum vertex multicut size of the set of terminal pairs $T$.
  \item {\ProblemNameShort} is W[2]-complete parameterized by~$k$ even if $T = \binom{Q}{2}$ for some $Q \subseteq V$.
\end{enumerate}

The first one implicitly has~$k$ in the parameter since such a cut size is an upper bound of~$k$.
Also, it is a lower bound of the minimum vertex multiway-cut size of $V(T)$.
In this sense, this intractability result tightens the gap between the W[2]-completeness parameterized solely by~$k$~\cite{AravindS24}
and our fixed-parameter tractability result parameterized by the minimum vertex multiway-cut size of $V(T)$.

The second result shows that a natural restriction of {\ProblemNameShort} is still W[2]-complete parameterized solely by~$k$,
where we want to hit pairwise geodesic intervals among a given vertex subset $Q \subseteq V$.
This result is shown by a modification of the proof of the first result.

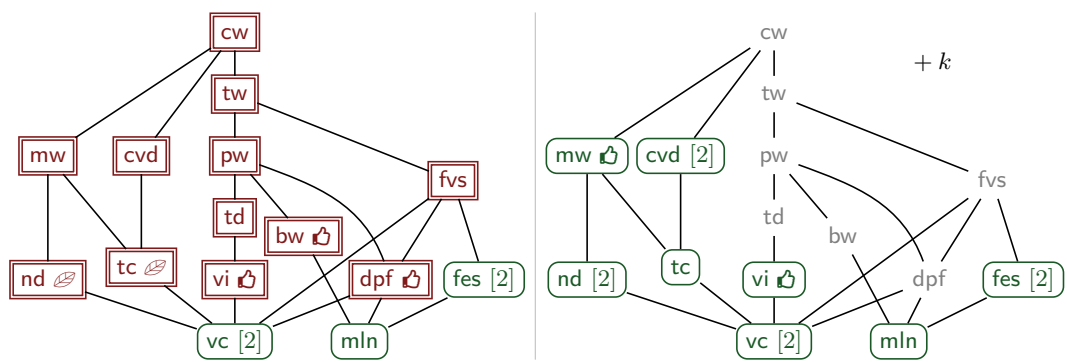
\begin{figure}[t]
    \centering


\definecolor{tkzdarkred}{rgb}{.5,.1,.1}
\definecolor{tkzdarkorange}{rgb}{.7,.25,0}
\definecolor{tkzdarkblue}{rgb}{.1,.1,.5}
\definecolor{tkzdarkgreen}{rgb}{.1,.35,.15}

\tikzset{npc/.style = {draw,semithick,rectangle,double,tkzdarkred,align=center}}
\tikzset{whd/.style = {draw,semithick,rectangle,tkzdarkorange,align=center}}
\tikzset{fpt/.style = {draw,semithick,rectangle,rounded corners,tkzdarkgreen,align=center}}
\tikzset{unk/.style = {gray,align=center}}

\begin{tikzpicture}[semithick, every node/.style={font=\small,text height=1.3ex, text depth = 0.2ex},scale=0.82]
  \node[npc] (cw) at (0cm, 5.0cm) {$\cw$};
  \node[npc] (tw) at (0cm, 4.0cm) {$\tw$};
  \node[npc] (pw) at (0cm, 3.0cm) {$\pw$};
  \node[npc] (td) at (0cm, 2.0cm) {$\td$};
  \node[npc] (vi) at (0cm, 1.0cm) {$\vi$\thispaper};
  \node[fpt] (vc) at (0cm, 0.0cm) {$\vc$~\cite{AravindS24}};
  \node[npc] (mw) at (-3cm, 3.0cm) {$\mw$};
  \node[npc] (nd) at (-3cm, 1.0cm) {$\nd$\easyobs};
  \node[npc] (cvd) at (-1.5cm, 3cm) {$\cvd$};
  \node[npc] (tc) at (-1.5cm, 1.2cm) {$\tc$\easyobs};
  \node[npc] (fvs) at (3.5cm, 2.6cm) {$\fvs$};
  \node[fpt] (fes) at (4cm, 1.0cm) {$\fes$~\cite{AravindS24}};
  \node[fpt] (mln) at (2cm, 0.0cm) {$\mln$};
  \node[npc] (bw) at (1.1cm, 1.7cm) {$\bw$\thispaper};
  \node[npc] (dpf) at (2.5cm, 1.0cm) {$\dpf$\thispaper};

  \draw (cw) -- (mw) -- (nd) -- (vc);
  \draw (cw) -- (cvd) -- (tc) -- (vc);
  \draw (mw) -- (tc);
  \draw (cw) -- (tw) -- (pw) -- (td) -- (vi) -- (vc);
  \draw (tw) -- (fvs) -- (vc);
  \draw (fvs) -- (fes) -- (mln);
  \draw (pw) -- (bw) -- (mln);
  \draw (pw) to [out=-20, in=110] (dpf) -- (vc);
  \draw (fvs) -- (dpf) -- (mln);
  
\end{tikzpicture} 
\hfill
\textcolor{lightgray}{\vrule}
\hfill
\begin{tikzpicture}[semithick, every node/.style={font=\small,text height=1.3ex, text depth = 0.2ex},scale=0.82]
  \node[unk] (cw) at (0cm, 5.0cm) {$\cw$};
  \node[unk] (tw) at (0cm, 4.0cm) {$\tw$};
  \node[unk] (pw) at (0cm, 3.0cm) {$\pw$};
  \node[unk] (td) at (0cm, 2.0cm) {$\td$};
  \node[fpt] (vi) at (0cm, 1.0cm) {$\vi$\thispaper};
  \node[fpt] (vc) at (0cm, 0.0cm) {$\vc$~\cite{AravindS24}};
  \node[fpt] (mw) at (-3cm, 3.0cm) {$\mw$\thispaper};
  \node[fpt] (nd) at (-3cm, 1.0cm) {$\nd$~\cite{AravindS24}};
  \node[fpt] (cvd) at (-1.5cm, 3cm) {$\cvd$~\cite{AravindS24}};
  \node[fpt] (tc) at (-1.5cm, 1.2cm) {$\tc$};
  \node[unk] (fvs) at (3.5cm, 2.6cm) {$\fvs$};
  \node[fpt] (fes) at (4cm, 1.0cm) {$\fes$~\cite{AravindS24}};
  \node[fpt] (mln) at (2cm, 0.0cm) {$\mln$};
  \node[unk] (bw) at (1.1cm, 1.7cm) {$\bw$};
  \node[unk] (dpf) at (2.5cm, 1.0cm) {$\dpf$};

  \draw (cw) -- (mw) -- (nd) -- (vc);
  \draw (cw) -- (cvd) -- (tc) -- (vc);
  \draw (mw) -- (tc);
  \draw (cw) -- (tw) -- (pw) -- (td) -- (vi) -- (vc);
  \draw (tw) -- (fvs) -- (vc);
  \draw (fvs) -- (fes) -- (mln);
  \draw (pw) -- (bw) -- (mln);
  \draw (pw) to [out=-20, in=110] (dpf) -- (vc);
  \draw (fvs) -- (dpf) -- (mln);

  \node (plusk) at (2.5cm, 4.5cm) {${}+k$};
\end{tikzpicture} 

    \caption{The figure on the left shows the complexity of {\ProblemName} parameterized solely by each parameter,
      and the one on the right shows the complexity parameterized by each parameter plus solution size~$k$.
      (See \cref{sec:preliminaries} for the non-abbreviated names of the parameters.)
      The results marked with{\thispaper} are shown in this paper
      and the ones with{\easyobs} are easy observations.
      The double rectangles and the rounded rectangles indicate paraNP-complete and fixed-parameter tractable cases, respectively. 
      Ones without frames remain unsettled.
      A connection between two parameters means that if the one below (say $\beta$) is bounded, then so is the one above (say $\alpha$);
      that is, there is a function $f$ such that $\alpha(G) \le f(\beta(G))$ for every graph $G$.
      Thus, a positive result propagates downward, while a negative result propagates upward.
    }

    \label{fig:graph-parameters}
\end{figure}


\subsection{Related work}
In the fields of graph theory and graph algorithms, 
the concept of geodesic intervals (also known as \emph{geodetic intervals} or just as \emph{intervals}) is well studied,
especially in the context of the problem \textsc{Geodetic Set}~\cite{HararyLT93} (see the survey on \textsc{Geodetic Set} \cite{BresarKT11} and the references therein).
In a graph $G = (V,E)$, a set $S \subseteq V$ is a \emph{geodetic set} if $\bigcup_{\{u,v\} \in \binom{S}{2}} I_{G}[u,v] = V$;
that is, the union of pairwise geodesic intervals among the vertices in $S$ covers all vertices of $G$.
Given a graph $G$ and an integer $k$, \textsc{Geodetic Set} asks whether $G$ has a geodetic set of size at most~$k$.
\textsc{Geodetic Set} is NP-complete~\cite{Atici02,DoughatK96}
and several parameterized complexity results have been shown recently~\cite{KellerhalsK22,Tale25arXiv}.
Although {\ProblemName} and \textsc{Geodetic Set} share some concepts in common, there seem to be no direct connections. 



\section{Preliminaries}
\label{sec:preliminaries}

Throughout the paper, we assume that the reader is familiar with the concepts in parameterized complexity and fixed-parameter tractability
(see standard textbooks, e.g., \cite{CyganFKLMPPS15,DowneyF99,DowneyF13,FlumG06,Niedermeier06}).

All graphs considered in this paper are finite, simple, and undirected.
In \cref{sec:fpt-mwc-vi}, we consider edge-weighted graphs, where the weight of each edge represents its length.
To guarantee the existence of shortest paths, we assume that the input graph does not contain a cycle of a negative total weight.
It is well known that even in the edge-weighted setting, we can compute all-pairs shortest paths (or find a negative cycle) in polynomial time
using standard algorithms (see e.g., \cite{CormenLRS22}).

As mentioned in the introduction, {\WProblemNameShort} is a special case of \textsc{Hitting Set}.
Given a finite set $U$, a family $\mathcal{F} \subseteq 2^{U}$, and an integer $k$,
\textsc{Hitting Set} asks whether there is $S \subseteq U$ with $|S| \le k$ such that $S \cap F \ne \emptyset$ for each $F \in \mathcal{F}$.
We can see that an instance $\langle G = (V,E), T, k \rangle$ of {\WProblemNameShort} is equivalent to
the instance $\langle V, \{I_{G}[u,v] \mid \{u,v\} \in T\}, k \rangle$ of \textsc{Hitting Set},
where the family $\{I_{G}[u,v] \mid \{u,v\} \in T\}$ can be computed in polynomial time.
\textsc{Hitting Set} is known to be W[2]-complete parameterized by $k$,
while it is fixed-parameter tractable parameterized by $|U|$ and by $|\mathcal{F}|$ (see \cite{CyganFKLMPPS15}).

Although our results involve many structural graph parameters as shown in \cref{fig:graph-parameters}, 
not all of the parameters need their formal definitions for the presentations in this paper.
Thus, we only define some of them when it is necessary.
For the definitions and relationships of the parameters,
we refer the reader to the survey on structural graph parameters~\cite{SorgeW19} (see also \cite{Schroder19,Tran22}).

In \cref{fig:graph-parameters}, we abbreviated the names of graph parameters as follows:
$\cw$ is cliquewidth;
$\tw$ is treewidth; 
$\pw$ is pathwidth; 
$\td$ is treedepth; 
$\vi$ is vertex integrity;
$\vc$ is vertex cover number;
$\mw$ is modular-width;
$\nd$ is neighborhood diversity;
$\cvd$ is cluster vertex deletion number;
$\tc$ is twin cover number;
$\bw$ is bandwidth;
$\mln$ is max-leaf number;
$\fvs$ is feedback vertex set number;
$\fes$ is feedback edge set number; and 
$\dpf$ is distance to path forest.


\section{NP-completeness on highly restricted graphs}
\label{sec:npc}

Since {\ProblemNameShort} clearly belongs to NP, we only prove the NP-hardness in the following proofs.
\begin{theorem}
\label{thm:path-forest}
{\ProblemName} is $\mathrm{NP}$-complete on graphs of vertex-deletion distance~$1$ to the class of disjoint unions of $5$-vertex paths.
\end{theorem}

\begin{proof}
We present a polynomial-time reduction from \textsc{3-Coloring}, which is NP-complete~\cite[GT4]{GareyJ79}.
Given a graph $G = (V,E)$, \textsc{3-Coloring} asks whether there is a \emph{proper $3$-coloring} of $G$,
i.e., a mapping $c \colon V \to \{1,2,3\}$ such that $c(u) \ne c(v)$ for every $\{u,v\} \in E$.

Let $G = (V,E)$ be an instance of \textsc{3-Coloring}.
In the following, we construct an equivalent instance $\langle H, T, k \rangle$ of {\ProblemName}.

\proofsubparagraph{Construction of $\langle H, T, k \rangle$.}
We set $k = 2|V|$.
For each vertex $v \in V$, we construct a path $P_{v}$ of five vertices
$v_{1}$, $v_{1}'$, $v_{2}$, $v_{3}'$, and $v_{3}$ that appear in this order along the path.
To complete the construction of $H$, 
we add one more vertex, $s^{*}$, that is adjacent to almost all other vertices except for the center vertex $v_{2}$ in each path $P_{v}$ for each $v \in V$.
See \cref{fig:path-forest}.
Clearly, $H - s^{*}$ is a disjoint union of $5$-vertex paths.
For each $v \in V$, we add pairs $\{v_{1}, v_{1}'\}$ and $\{v_{3}, v_{3}'\}$ into~$T$.
For each $\{u,v\} \in E$, we add pairs $\{u_{1}, v_{1}\}$, $\{u_{2}, v_{2}\}$, and $\{u_{3}, v_{3}\}$ into~$T$.

\begin{figure}[tbh]
  \centering
  \includegraphics{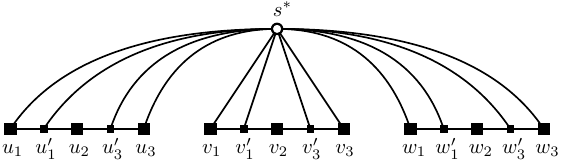}
  \caption{The construction of $H$ in \cref{thm:path-forest} (when $G$ has only three vertices $u$, $v$, $w$).
    We represent terminal vertices by squares.
    Some terminal vertices are depicted as small squares since they appear only in ``local'' terminal pairs.}
  \label{fig:path-forest}
\end{figure}

\proofsubparagraph{If $G$ is a yes instance.}
Let $c \colon V \to \{1, 2, 3\}$ be a proper 3-coloring of $G$.
For each $v \in V$, let $S_{v} \subseteq \{v_{1}, v_{1}', v_{2}, v_{3}', v_{3}\}$ be the following set depending on~$c(v)$:
\[
  S_{v} = 
  \begin{cases}
    \{v_{1}', v_{3}\} & \textrm{if } c(v) = 1, \\
    \{v_{1}, v_{3}\}  & \textrm{if } c(v) = 2, \\
    \{v_{1}, v_{3}'\} & \textrm{if } c(v) = 3. \\
  \end{cases}
\]
We set $S = \bigcup_{v \in V} S_{v}$. 
Since $|S| = 2|V| = k$, it suffices to show that $S$ interrupts all pairs in $T$.
For each $v \in V$, it is easy to see that $S_{v}$ interrupts the pairs $\{v_{1}, v_{1}'\}$ and $\{v_{3}, v_{3}'\}$.

For each edge $\{u,v\} \in E$, we show that $S_{u} \cup S_{v}$ interrupts the pairs $\{u_{1}, v_{1}\}$, $\{u_{2}, v_{2}\}$, and $\{u_{3}, v_{3}\}$.
Since $c$ is a proper 3-coloring, we have $c(u) \ne c(v)$, and thus $\{c(u), c(v)\} \cap \{1,3\} \ne \emptyset$.
By symmetry, we may assume that $c(u) \in  \{1,3\}$. Assume that $c(u) = 1$.
Then, $S_{u} = \{u_{1}', u_{3}\}$ interrupts $\{u_{3}, v_{3}\}$
and also $\{u_{2}, v_{2}\}$ as $u_{1}' \in I_{H}[u_{2}, v_{2}]$.
Since $c(v) \in \{2,3\}$, $S_{v}$ includes $v_{1}$, and thus it interrupts $\{u_{1}, v_{1}\}$.
The other case of $c(u) = 3$ can be handled in the same way by swapping $1$ and $3$ in the discussion above.

\proofsubparagraph{If $\langle H, T, k \rangle$ is a yes instance.}
Let $S \subseteq V(H)$ be a set with $|S| \le k$ that interrupts all pairs in $T$.
Since $\{v_{1}, v_{1}'\}, \{v_{3}, v_{3}'\} \in E(H) \cap T$ for every $v \in V$ and $k = 2|V|$, we have
$|S \cap \{v_{1}, v_{1}'\}| = |S \cap \{v_{3}, v_{3}'\}| = 1$ for each $v \in V$
and
$S \subseteq \bigcup_{v \in V} \{v_{1}, v_{1}', v_{3}, v_{3}'\}$. 
Let $S_{v} = S \cap \{v_{1}, v_{1}', v_{3}, v_{3}'\}$.
We define $c(v) \in \{1,2,3\}$ for each $v \in V$ as follows:
\[
  c(v) =
  \begin{cases}
    1 & \textrm{if } v_{1} \notin S_{v}, \\
    3 & \textrm{else if } v_{3} \notin S_{v}, \\
    2 & \textrm{otherwise } (S_{v} = \{v_{1}, v_{3}\}). \\
  \end{cases}
\]
We show that $c$ is a proper 3-coloring of $G$.
Suppose to the contrary that there is an edge $\{u,v\} \in E$ with $c(u) = c(v) = x$.
If $x \in \{1,3\}$, then $u_{x}, v_{x} \notin S_{u} \cup S_{v}$.
Since $I_{H}[u_{x}, v_{x}] = \{u_{x}, s^{*}, v_{x}\}$ and $s^{*} \notin S$,
$S$ does not interrupt the pair $\{u_{x}, v_{x}\} \in T$.
If $x = 2$, then $S_{u} \cup S_{v} = \{u_{1}, u_{3}, v_{1}, v_{3}\}$ does not intersect 
$I_{H}[u_{2}, v_{2}] = \{u_{2}, u_{1}', u_{3}', s^{*}, v_{1}', v_{3}', v_{2}\}$,
and thus $S$ ($\not\owns s^{*}$) does not interrupt the pair $\{u_{2}, v_{2}\} \in T$.
\end{proof}

To the graph $H$ constructed in the proof of \cref{thm:path-forest}, we may add some vertices and edges without invalidating the reduction
as long as the geodesic intervals to be hit do not change.
In particular, we can pick arbitrary two vertices $u$ and $v$ of $G$ 
and add a long path between $u_{3}$ and $v_{1}$ to $H$.\footnote{%
Actually, we can show that adding an edge between $u_{3}$ and $v_{1}$ is good enough.}
Repeating this in an appropriate way, we can make $H - s^{*}$ a path and obtain the following corollary.
\begin{corollary}
\label{cor:path}
{\ProblemName} is $\mathrm{NP}$-complete on graphs of vertex-deletion distance~$1$ to the class of paths.
\end{corollary}

The next theorem can be shown by replacing $5$-vertex paths in the proof of \cref{thm:path-forest} with triangles.
(The proofs of the claims marked with {\mainbodyrepeatedtheorem} are deferred to the appendix.)
\begin{theoremrep}
\label{thm:triangle}
{\ProblemName} is $\mathrm{NP}$-complete on graphs of vertex-deletion distance~$1$ to the disjoint unions of triangles.
\end{theoremrep}
\begin{proof}
Let $G = (V,E)$ be an instance of \textsc{3-Coloring}.
We set $k = 2|V|$.
For each $v \in V$, the graph $H$ contains a triangle $C_{v}$ with vertices $v_{1}$, $v_{2}$, and $v_{3}$.
Additionally, $H$ has vertex $s^{*}$ adjacent to all other vertices. See \cref{fig:triangle-bw}~(left).
For each $v \in V$, we add pairs $\{v_{1}, v_{2}\}$, $\{v_{2}, v_{3}\}$, and $\{v_{3}, v_{1}\}$ into~$T$.
For each $\{u,v\} \in E$, we add pairs $\{u_{1}, v_{1}\}$, $\{u_{2}, v_{2}\}$, and $\{u_{3}, v_{3}\}$ into~$T$.

Assume that $G$ admits a proper 3-coloring $c$.
For each $v \in V$, let $S_{v} = \{v_{1}, v_{2}, v_{3}\} \setminus \{v_{c(v)}\}$.
For each $v \in V$, $S_{v}$ interrupts the pairs $\{v_{1}, v_{2}\}$, $\{v_{2}, v_{3}\}$, and $\{v_{3}, v_{1}\}$.
For each edge $\{u,v\} \in E$, since $c(u) \ne c(v)$, 
$S_{u} \cup S_{v}$ includes at least one of $u_{i}$ and $v_{i}$ for each $i \in \{1,2,3\}$,
and thus, $S_{u} \cup S_{v}$ interrupts the pairs $\{u_{1}, v_{1}\}$, $\{u_{2}, v_{2}\}$, and $\{u_{3}, v_{3}\}$.
Since $\bigcup_{v \in V} S_{v}$ has size $2|V| = k$,
$\langle H, T, k \rangle$ is a yes instance of {\ProblemNameShort}\@.

Next assume that $\langle H, T, k \rangle$ is a yes instance of {\ProblemNameShort},
and that $S \subseteq V(H)$ interrupts all pairs in $T$ and has size at most $k$.
Since $\{v_{1}, v_{2}\}, \{v_{2}, v_{3}\}, \{v_{1}, v_{1}\} \in E(H) \cap T$ and $k = 2|V|$, we have
$s^{*} \notin S$ and $|S \cap \{v_{1}, v_{2}, v_{3}\}| = 2$ for each $v \in V$. 
We define $c(v) \in \{1,2,3\}$ for each $v \in V$ to be the unique index $i \in \{1,2, 3\}$ such that $v_{i} \notin S$.
Suppose that there is an edge $\{u,v\} \in E$ with $c(u) = c(v) = x$.
The definition of $c$ implies that $u_{x}, v_{x} \notin S \cap \{u_{1}, u_{2}, u_{3}, v_{1}, v_{2}, v_{3}\}$.
Since $s^{*} \notin S$ and $I_{H}[u_{x}, v_{x}] = \{u_{x}, s^{*}, v_{x}\}$, 
$S$ does not interrupt the pair $\{u_{x}, v_{x}\} \in T$.
\begin{figure}[tbh]
  \centering
  \includegraphics{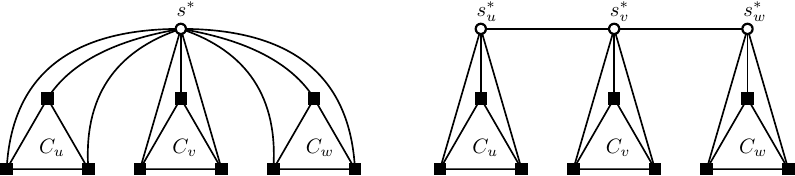}
  \caption{The construction of $H$ and $H'$ in \cref{thm:triangle,cor:bw}.}
  \label{fig:triangle-bw}
\end{figure}
\end{proof}

To show the hardness on graphs of bandwidth~$4$ and maximum degree~$5$, we only need to slightly modify the proof of \cref{thm:triangle}.
For an $n$-vertex graph $G = (V,E)$, its bandwidth $\bw(G)$ is defined as $\bw(G) = \min_{\sigma} \max_{\{u,v\} \in E} |\sigma(u) - \sigma(v)|$,
where the minimum is taken overall bijections $\sigma \colon V \to [n]$, where $[n] = \{i \in \mathbb{Z} \mid 1 \le i \le n\}$.
\begin{corollaryrep}
\label{cor:bw}
{\ProblemName} is $\mathrm{NP}$-complete on graphs of bandwidth~$4$ and maximum degree~$5$.
\end{corollaryrep}
\begin{proof}
Let $H$ be the graph constructed in the proof of \cref{thm:triangle}.
We first remove $s^{*}$ from $H$. 
Next, for each $v \in V$, we add a vertex $s^{*}_{v}$ adjacent to all vertices in the triangle $C_{v}$.
Finally, we add $|V|-1$ edges to make a path that goes through the new vertices $\{s^{*}_{v} \mid v \in V\}$.
See \cref{fig:triangle-bw}~(right).
Let us call the modified graph $H'$.
We can see that $\langle H, T, k \ (= 2|V|) \rangle$ and $\langle H', T, k \rangle$ are equivalent instances of {\ProblemNameShort}
since a solution $S$ of $\langle H', T, k \rangle$ cannot include any vertex in $\{s^{*}_{v} \mid v \in V\}$.
Clearly, $H'$ has maximum degree~$5$.
To see that $H'$ has bandwidth~$4$, consider the ordering that places the vertices along the path on $\{s^{*}_{v} \mid v \in V\}$,
while the vertices of $C_{v}$ are placed right after $s^{*}_{v}$.
This ordering has bandwidth~$4$. See \cref{fig:bw4}.
\begin{figure}[tbh]
  \centering
  \includegraphics{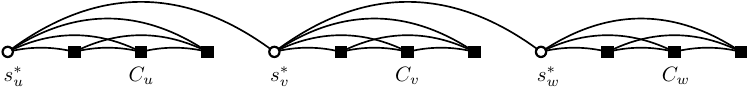}
  \caption{The embedding of $H'$ in \cref{cor:bw}.}
  \label{fig:bw4}
\end{figure}
\end{proof}



\section{Fixed-parameter tractability with $k$ plus structural parameters}
\label{sec:fpt}


\subsection{Modular-width}
Let $G = (V, E)$ be a graph.
A set $M \subseteq V$ is a \emph{module} of $G$ if $N_{G}(u) \setminus M = N_{G}(v) \setminus M$ holds for all $u,v \in M$.
In other words, each vertex $w \notin M$ is adjacent to all or no vertices of $M$.
Note that if $M$ and $M'$ are two disjoint modules of $G$, then either there are all or no possible edges between them.
Two modules $M$ and $M'$ are \emph{adjacent} if there are all possible edges between them, and \emph{nonadjacent} otherwise.

The \emph{modular-width} of $G = (V,E)$, denoted $\mw(G)$, is defined as the minimum integer $\mu$ that satisfies the following recursive condition:
\begin{itemize}
  \item $|V| \le \mu$; or
  \item $V$ can be partitioned into modules $M_{1}, \dots, M_{\mu'}$ of $G$
  such that $2 \le \mu' \le \mu$ and $\mw(G[M_{i}]) \le \mu$ for each $i \in [\mu']$.
\end{itemize}

\begin{theorem}
{\ProblemName} is fixed-parameter tractable parameterized by solution size~$k$ plus modular-width.
\end{theorem}
\begin{proof}
Let $\langle G = (V, E), T, k \rangle$ be an instance of {\ProblemNameShort}\@.
Observe that the vertex set of each connected component $C$ of $G$ is a module of $G$, and thus $\mw(C) \le \mw(G)$.
Hence, we may assume that $G$ is connected; otherwise we solve each connected component of $G$ independently for all $k' \le k$
and combine the answers for them.

Let $M_{1}, \dots, M_{p}$ be a partition of $V$ into modules of $G$ with $2 \le p \le \mw(G)$.
It is known that such a partition can be computed in linear time~\cite{McConnellS99}.
Note that, although the definition of modular-width requires a complete recursive structure,
our algorithm described below only needs a partition at the very first level.

For each $\{u,v\} \in T$, we take a set $J(u,v)$ and initialize it by $I_{G}[u,v]$.
Now we set $\mathcal{I} = \{J(u,v) \mid \{u,v\} \in T\}$.
Our goal is to find a hitting set $S$ of $\mathcal{I}$ with size at most $k$.
To this end, we make two phases of branching that make the rest of the problem almost trivially polynomial-time solvable.
Although we do not explicitly state the trivial termination conditions in each step,
we stop and reject a branch if it results in $k < 0$ or $\emptyset \in \mathcal{I}$.

\proofsubparagraph{The first branching.}
We first branch on the $2^{p}$ candidates of the set of modules that $S$ intersects, which we call $\mathcal{M}_{S}$.
We assume that $S$ is chosen in such a way that $|\mathcal{M}_{S}|$ is maximized.
According to $\mathcal{M}_{S}$, we update $\mathcal{I}$ as follows.
\begin{itemize}
  \item Add each $M_{i} \in \mathcal{M}_{S}$ into $\mathcal{I}$.

  \item For each $\{u,v\} \in T$, if there is $M_{i} \in \mathcal{M}_{S}$ with $M_{i} \subseteq J(u,v)$, then remove $J(u,v)$ from $\mathcal{I}$.

  \item For each $\{u,v\} \in T$ and $M_{i} \notin \mathcal{M}_{S}$, update $J(u,v)$ as $J(u,v) \coloneqq J(u,v) \setminus M_{i}$.
\end{itemize}

We now see that, for every $\{u, v\} \in T$, 
if $u$ and $v$ belong to different modules, say $M_{i}$ and $M_{j}$, then $J(u,v) \subseteq \{u,v\}$ after the update above.
Observe that $I_{G}[u,v] = \{u,v\} \cup \bigcup_{M \in \mathcal{M}} M$ for some $\mathcal{M} \subseteq \{M_{1}, \dots, M_{p}\} \setminus \{M_{i},M_{j}\}$.
Since $J(u,v) \in \mathcal{I}$, we have $\mathcal{M} \cap \mathcal{M}_{S} = \emptyset$, and thus $J(u,v) \subseteq \{u,v\}$.

\proofsubparagraph{The second branching.}

Let $\mathcal{I}'$ be the subset of $\mathcal{I}$ consisting of the sets $J(u,v)$ such that $\{u,v \} \in E$ or $u$ and $v$ belong to different modules.
Observe that every hitting set $S$ of $\mathcal{I}$ with size at most $k$ contains a \emph{minimal} hitting set $S'$ of $\mathcal{I}'$ with size at most $k$.
We list all possible candidates for $S'$ and branch on them.
Since each member of $\mathcal{I}'$ has size at most~$2$, 
we can use the algorithm of Damaschke~\cite{Damaschke06} for enumerating all minimal vertex covers of size at most~$k$
in $O(|\mathcal{I}'| + k^{2} 2^{k})$ time. It gives us all (at most $2^{k}$) minimal hitting sets of $\mathcal{I}'$ with size at most $k$ (or answers that there is no such a set).
After we pick one candidate $S'$, we decrease $k$ by $|S'|$ and update $\mathcal{I}$ by removing all elements hit by $S'$.
In particular, we remove all $J(u,v) \in \mathcal{I}'$ from $\mathcal{I}$.

\proofsubparagraph{Solving the reduced instance in polynomial time.}
We first show that, at this point, each member of $\mathcal{I}$ is a subset of some module $M_{i}$.
This is trivial for the members of $\mathcal{M}_{S}$.
For $J(u,v) \in \mathcal{I}$, we have $\{u,v\} \notin E$ and there is a module $M_{i}$ such that $u, v \in M_{i}$.
This implies that $\dist_{G}(u,v) = 2$ since $M_{i}$ has an adjacent module by the connectivity of~$G$.
Thus, $I_{G}[u,v]$ consists of the modules adjacent to $M_{i}$ and a subset of $M_{i}$.
Since $J(u,v) \in \mathcal{I}$ now, the adjacent modules do not belong to $\mathcal{M}_{S}$.
This implies that $J(u,v) \subseteq M_{i}$.

For $i \in [p]$, let $\mathcal{I}_{i}$ be the (possibly empty) subset of $\mathcal{I}$ consisting of the members that are subsets of $M_{i}$.
Since $\mathcal{I}_{1}, \dots, \mathcal{I}_{p}$ form a partition of $\mathcal{I}$ with disjoint universes, we can handle them separately.
That is, if we know a minimum hitting set $S_{i}$  of $\mathcal{I}_{i}$ for each $i \in [p]$,
then we can set $S = S' \cup \bigcup_{i \in [p]} S_{i}$.

We now claim that, under the assumption that $\mathcal{M}_{S}$ and $S'$ are correct ones, each nonempty $\mathcal{I}_{i}$ admits a hitting set of size~$1$.
Suppose to the contrary that a minimum hitting set $S_{i}$ of $\mathcal{I}_{i}$ has size at least $2$ for some $i \in [p]$.
Let $M_{j}$ be a module adjacent to $M_{i}$, and let $v_{i} \in S_{i}$ and $v_{j} \in M_{j}$.
Let $S'_{i} = (S_{i} \setminus \{v_{i}\}) \cup \{v_{j}\}$ and $R = (S \setminus S_{i}) \cup S'_{i}$.
Observe that $S'_{i}$ hits $\mathcal{I}_{i}$: it contains at least one vertex of $M_{i}$, which hits $M_{i}$;
and $v_{j} \in M_{j}$ hits all $J(u,v) \in \mathcal{I}_{i}$ since $\dist_{G}(u,v) = 2$.
This implies that $R$ interrupts all pairs in $T$ as $S$ does.
However, this contradicts the assumption on $S$ since $R$ intersects one more module ($M_{j}$) than $S$.
Thus, we can conclude that at least one of $\mathcal{M}_{S}$ and $S'$ is not correct.

Given the observation above, we can solve the rest of the current branch in polynomial time as follow.
If there are at most $k$ indices $i \in [p]$ such that $\mathcal{I}_{i}$ is nonempty
and each of them admits a hitting set of size~$1$, then return yes.
Otherwise, we reject the current branch.

\proofsubparagraph{Total running time.}
We first branch on $2^{p}$ cases.
In each case, we take $O^{*}(2^{k})$ time to enumerate at most $2^{k}$ candidates for the next branch.
We branch on these candidates and take polynomial time in each of them.
In total, the running time is $O^{*}(2^{p+k})$.
\end{proof}


\subsection{Vertex integrity and vertex multiway cut}
\label{sec:fpt-mwc-vi}
To show that {\WProblemNameShort} is fixed-parameter tractable parameterized by~$k$ plus vertex integrity
and by the minimum vertex multiway-cut size of terminal vertices $V(T)$,
we first give a single algorithm that applies to a more general setting that includes both cases
under the assumption that a certain kind of a separator is given as part of the input.
We then present simple observations of how the requirement that the separator is given as part of the input can be dropped for our cases.

\subsubsection{The main algorithm}

Let us first give an overview of the main algorithm.
It can be seen as a generalization of the one parameterized by vertex cover number
due to Aravind and Saxena~\cite{AravindS24}.
Instead of a vertex cover, we use a small set $Z \subseteq V$ such that 
each connected component of $G - Z$ contains a small number of terminal vertices.
Given $Z$, we guess how our solution $S$ interrupts pairs in $\binom{Z}{2}$.
To reflect the guess, we then update the family $\mathcal{I}$ of vertex sets to be hit by the solution 
(which was originally the family of geodesic intervals for $T$), in an appropriate way.
We show that, after this update, each member of $\mathcal{I}$ is either a geodesic interval between two vertices in $Z$
or a vertex set intersecting at most two connected components of $G-Z$.
We then show that although each connected component of $G-Z$ may have unbounded size,
we can focus on a bounded number of ``representatives'' in each connected component.
This allows us to shrink the members of $\mathcal{I}$ and then 
apply the sunflower lemma to a subfamily of $\mathcal{I}$.
As a result, we get an upper bound of $|\mathcal{I}|$, and thus the rest becomes easy to solve.

\begin{theorem}
\label{thm:whgi-p-q-sep}
Consider a special case of  {\WProblemNameShort} where, in addition to $\langle G = (V,E), T, k\rangle$,
we are given a set $Z \subseteq V$ such that $|Z| \le p$ and each connected component of $G - Z$ contains at most $q$ terminal vertices.
There is a fixed-parameter algorithm for this problem parameterized by $k+p+q$.
\end{theorem}
\begin{proof}
For simplicity, we assume that $Z$ does not include any terminal vertex.
This assumption can be justified by the following simple modification:
if a vertex $z \in Z$ is involved in terminal pairs,
then add a new vertex $z'$ adjacent only to $z$ and replace all terminal pairs $\{z, x\}$ with $\{z', x\}$;
this modification is safe as any path to $z'$ must go through $z$.
We also assume that our solution does not use any vertex in $Z$.
For this assumption, we guess from at most $2^{p}$ candidates the subset $Z'$ of $Z$ that actually is contained in the solution,
remove the terminal pairs interrupted by $Z'$, and reduce $k$ by $|Z'|$.

Let $u, v \in V \setminus Z$ and $C_{u}$ and $C_{v}$ be the connected components of $G - Z$ that contain $u$ and $v$, respectively.
Note that $C_{u}$ and $C_{v}$ are not necessarily different. We now show that
\begin{equation}
  I_{G}[u,v] \setminus \left(Z \cup \textstyle\bigcup_{\{z, z'\} \in \binom{Z}{2}} I_{G}[z,z']\right) \subseteq V(C_{u}) \cup V(C_{v}).
  \label{eq:vi-sep}
\end{equation}
Let $P$ be a shortest $u$--$v$ path.
To show \eqref{eq:vi-sep}, it suffices to show that either $V(P) \subseteq V(C_{u}) \cup V(C_{v})$ or
$V(P) \setminus V(P') \subseteq V(C_{u}) \cup V(C_{v})$ for some shortest path $P'$ connecting (possibly the same) vertices $z, z' \in Z$.
If $P$ does not visit any vertex in $Z$, then $V(P) \subseteq V(C_{u})$ holds.
Assume that $P$ visits some vertices in $Z$.
Let $P_{Z}$ be the maximal subpath of $P$ such that the first and last vertices belong to $Z$. 
Note that $P_{Z}$ is a shortest path.
Let $P_{u}$ and $P_{v}$ be the paths containing $u$ and $v$, respectively, obtained from $P$ by removing $V(P_{Z})$.
Since $Z$ separates $C_{u}$ from the rest of the graph, $V(P_{u}) \subseteq V(C_{u})$ holds.
By the same argument, we can see that $V(P_{v}) \subseteq V(C_{v})$.
Hence, $V(P) \setminus V(P_{Z}) \subseteq V(C_{u}) \cup V(C_{v})$ holds.
This implies~\eqref{eq:vi-sep}.

\proofsubparagraph{Branching.}
For each $\{u,v\} \in T$, we take a set $J(u,v)$ and initialize it by $I_{G}[u,v] \setminus Z$.
(Recall that we already guessed and preprocessed the intersection of $Z$ and the solution.)
Now we set $\mathcal{I} = \{J(u,v) \mid \{u,v\} \in T\}$.
We branch on the $2^{\binom{|Z|}{2}}$ candidates of the subset of $\binom{Z}{2}$ that $S$ interrupts, which we call $Y$.
That is, we assume that we are looking for a solution $S$ such that $S \cap I_{G}[u,v] \ne \emptyset$ for each $\{u,v\} \in Y$, 
and $S \cap I_{G}[u,v] = \emptyset$ for each $\{u,v\} \in \binom{Z}{2} \setminus Y$.
Thus, we update $\mathcal{I}$ as follows.
\begin{itemize}
  \item For each $\{u,v\} \in Y$, let $J(u,v) = I_{G}[u,v] \setminus Z$ and add $J(u,v)$ into $\mathcal{I}$.

  \item For each $\{u,v\} \in T$ and $\{u',v'\} \in Y$,
  if $J(u',v') \subseteq J(u, v)$, then remove $J(u,v)$ from $\mathcal{I}$.

  \item For each $\{u,v\} \in T$ and $\{u',v'\} \in \binom{Z}{2} \setminus Y$,
  if $J(u',v') \subseteq J(u, v)$,
  then update $J(u,v)$ as $J(u,v) \coloneqq J(u,v) \setminus J(u',v')$.
\end{itemize}
Let $\mathcal{I}_{T} = \{J(u,v) \in \mathcal{I} \mid \{u,v\} \in T\}$ and 
$\mathcal{I}_{Y} = \{J(u,v) \in \mathcal{I} \mid \{u,v\} \in Y\}$.
The condition~\eqref{eq:vi-sep} implies that,
after the update above, $J(u,v) \subseteq V(C_{u}) \cup V(C_{v})$ holds for each $J(u,v) \in \mathcal{I}_{T}$.

\proofsubparagraph{Reducing the size of each set in $\mathcal{I}_{T}$.}
Two vertices $u$ and $v$ are \emph{$T$-equivalent} if $u$ and $v$ interrupt the same set of pairs in $T$.
In other words, $u$ and $v$ are $T$-equivalent if, for each $\{x,y\} \in T$, 
either $\{u,v\} \subseteq I_{G}[x,y]$ or $\{u,v\} \cap I_{G}[x,y] = \emptyset$ holds.
Observe that the $T$-equivalence among vertices is an equivalence relation.
Every inclusion-wise minimal solution for $\langle G, T, k \rangle$ includes at most one vertex of a $T$-equivalence class,
and if it includes one vertex of a $T$-equivalence class, then taking any of them is equivalent.
The partition of the vertices to the $T$-equivalence classes can be computed in polynomial time by first computing all-pairs shortest distance in $G$,
and then iteratively refining the partition of the vertices with respect to the interruption relation to each pair $\{u,v\} \in T$.

Let $V'$ be a subset of $V$ constructed by picking one vertex from each $T$-equivalence class.
Using $V'$, we further update $\mathcal{I}$ by compressing each $J(u,v) \in \mathcal{I}$ as $J(u,v) \coloneqq J(u,v) \cap V'$.
We now show that this update allows us to upper-bound the size of each $J(u,v) \in \mathcal{I}_{T}$ by a function depending only on $p+q$.
To this end, we prove the following claim.

\begin{claim}
\label{clm:component-classes}
Let $C$ be a connected component of $G - Z$.
If two vertices $x, y \in V(C)$ are $\binom{Z \cup (V(T) \cap V(C))}{2}$-equivalent,
then they are $T$-equivalent.
\end{claim}
\begin{claimproof}
It suffices to show that if one of $x$ and $y$ interrupts a pair $\{a,b\} \in T$, then the other also interrupts $\{a,b\}$.
By symmetry, we may assume that $x$ interrupts $\{a,b\}$.

Let $P_{a x}$ be a shortest $a$--$x$ path and $P_{x b}$ a shortest $x$--$b$ path.
Since $x$ interrupts $\{a,b\}$, the concatenation $P_{a b}^{(x)} \coloneqq P_{a x} P_{x b}$ is a shortest $a$--$b$ path.
Now we pick two vertices $a' \in V(P_{a x})$ and $b' \in V(P_{x b})$ as follows:
we set $a' = a$ if $a \in V(C)$; otherwise, we set $a'$ to an arbitrary vertex in $V(P_{a x}) \cap Z$
(such a vertex exists as $Z$ separates $C$ from the rest of the graph);
analogously, we set $b' = b$ if $b \in V(C)$; otherwise, we set $b'$ to an arbitrary vertex in $V(P_{x b}) \cap Z$.
By the definition, $\{a', b'\} \subseteq Z \cup (V(T) \cap V(C))$ holds.

Let $P_{a' b'}^{(x)}$ be the subpath of $P_{a b}^{(x)}$ from $a'$ to $b'$. 
Since $P_{a' b'}^{(x)}$ is a shortest $a'$--$b'$ path containing $x$, $x$ interrupts $\{a', b'\}$,
and thus so does $y$ by their $\binom{Z \cup (V(T) \cap V(C))}{2}$-equivalence.
Thus, there is a shortest $a'$--$b'$ path $P_{a' b'}^{(y)}$ that passes through $y$.
This implies that we can obtain a shortest $a$--$b$ path that visits $y$
by replacing $P_{a' b'}^{(x)}$ in $P_{a b}^{(x)}$ with $P_{a' b'}^{(y)}$.
Hence, $y$ interrupts $\{a,b\}$.
\end{claimproof}

For every connected component $C$ of $G-Z$, the assumption on $Z$ implies that $|Z \cup (V(T) \cap V(C))| \le p+q$. 
Thus, \cref{clm:component-classes} implies that $C$ intersects at most $2^{\binom{|Z \cup (V(T) \cap V(C))|}{2}} \le 2^{\binom{p+q}{2}}$ $T$-equivalence classes,
and thus, $|V' \cap V(C)| \le 2^{\binom{p+q}{2}}$.
This implies that $|J(u,v)| \le 2^{\binom{p+q}{2} + 1}$ for each $J(u,v) \in \mathcal{I}_{T}$
since $J(u,v) \subseteq V' \cap (V(C_{u}) \cup V(C_{v}))$.

\proofsubparagraph{Reducing the size of $\mathcal{I}$.}
Now we shrink the size of $\mathcal{I}_{T}$ by a well-known application of the sunflower lemma~\cite{ErdosR1960} to $d$-\textsc{Hitting Set}, 
a variant of \textsc{Hitting Set} with each set having size at most $d$.
There is a polynomial-time algorithm based on the sunflower lemma that, given an instance $\mathcal{H}$ of $d$-\textsc{Hitting Set}, either 
\begin{itemize}
  \item answers that $\mathcal{H}$ admits no hitting set of size at most~$k$, or
  \item returns an instance $\mathcal{H}'$ of $d$-\textsc{Hitting Set} with at most $k^{d} \cdot d!$ sets such that 
  $\mathcal{H}'$ and $\mathcal{H}$ admit exactly the same (possibly empty) family of hitting sets of size at most~$k$.
\end{itemize}
See \cite[Section~9.1]{FlumG06} or \cite[Theorem~2.26]{CyganFKLMPPS15}.
By setting $d = 2^{\binom{p+q}{2} + 1}$, we apply this algorithm to $\mathcal{I}_{T}$ and obtain $\mathcal{I}_{T}'$ 
with at most $k^{d} \cdot d!$ sets (or find that it is a no-instance).
Finally, we set $\mathcal{I}' = \mathcal{I}_{T}' \cup \mathcal{I}_{Y}$.
The discussion so far guarantees that $\mathcal{I}'$ and $\mathcal{I}$ are equivalent.
Since $|\mathcal{I}'|$ ($= |\mathcal{I}_{T}'| + |\mathcal{I}_{Y}| \le k^{d} \cdot d! + p^{2}$) depends only on $k + p + q$,
we can find a minimum hitting set of $\mathcal{I'}$
by applying the standard fixed-parameter algorithm for \textsc{Hitting Set} parameterized by the number of sets (see \cref{sec:known-results}).
\end{proof}

\subsubsection{Applications of the main algorithm}
\label{sssec:app-main-alg}
\cref{thm:whgi-p-q-sep} requires a separator $Z$ with special properties as part of the input.
A natural question would be whether the original problem {\WProblemNameShort} is fixed-parameter tractable parameterized by $k + p + q$.
To this end, it suffices to show that finding $Z$ is fixed-parameter tractable parameterized by $p + q$.
Although we leave this general problem unsettled, we now observe that for our special cases finding $Z$ is fixed-parameter tractable.

The \emph{vertex integrity} of a graph $G = (V,E)$, denoted $\vi(G)$, is the minimum integer $\iota$ such that 
there is a set $X \subseteq V$ satisfying that $\iota = |X| + \max_{C \in \cc(G-X)} |V(C)|$,
where $\cc(G-X)$ is the set of connected components of $G-X$.
For an $n$-vertex graph $G$ with $\vi(G) = \iota$, 
finding a set $X \subseteq$ satisfying $\iota = |X| + \max_{C \in \cc(G-X)} |V(C)|$ can be done in time
$O(\iota^{\iota+1} n)$~\cite{DrangeDH16}.
Observe that such a set $X$ satisfies the conditions of $Z$ in \cref{thm:whgi-p-q-sep} with $p = q = \iota$.
Thus we get the following result as a direct corollary.
\begin{corollary}
\label{cor:vi}
{\WProblemName} is fixed-parameter tractable parameterized by solution size~$k$ plus vertex integrity. 
\end{corollary}

For a graph $G = (V,E)$ and $U \subseteq V$, a set $X \subseteq V$ is a \emph{vertex multiway cut} for $U$
if each connected component of $G - X$ contains at most one vertex of $U$.
It is known that the problem of finding a vertex multiway cut of size $r$ is fixed-parameter tractable
parameterized by~$r$~\cite{ChenLL09,Marx06}.
A multiway cut $X$ for $V(T)$ with size $r$ satisfies the conditions of $Z$ in \cref{thm:whgi-p-q-sep} with $p = r$ and $q = 1$.
Furthermore, $X$ interrupts all pairs in $T$ as no path can connect a pair in $T$ without passing through $X$.
Thus, we can assume that $k \le r$ as otherwise we already have a desired solution (i.e., $X$).
Now we apply \cref{thm:whgi-p-q-sep} and get the following result.
\begin{corollary}
\label{cor:multiway-cut}
{\WProblemName} is fixed-parameter tractable parameterized by the minimum vertex multiway-cut size of the set of terminal vertices $V(T)$.
\end{corollary}

As a final remark to this section, we note that the problem of finding $Z$ in \cref{thm:whgi-p-q-sep} is \emph{nonuniformly} fixed-parameter tractable parameterized by $p+q$.
This follows from the discussion by Marx~\cite[p.~396]{Marx06} on the case of $q = 1$ (i.e., the vertex multiway cut problem), which directly applies to general $q$ as well.
Given $G = (V, E)$ and $R \subseteq V$, we construct an edge-colored graph $G' = (V', E')$ as follows:
for each $v \in R$, add a vertex $v'$ and an edge $\{v, v'\}$;
color $E$ in black and $\{\{v, v'\} \mid v \in R\}$ in red.
Now,
there is $Z \subseteq V$ such that $|Z| \le p$ and each connected component of $G-Z$ contains at most $q$ terminal vertices
if and only if
there is $Z' \subseteq V'$ such that $|Z'| \le p$ and each connected component of $G'-Z'$ contains at most $q$ red edges.
Observe that the latter is a minor-closed property. 
More precisely, let $\mathcal{G}_{p,q}$ be the family of graphs with black and red edges such that 
each graph $H \in \mathcal{G}_{p,q}$ admits a separator $X \subseteq V(H)$
such that $|X| \le p$ and each connected component of $H - X$ contains at most $q$ red edges.
Then, $\mathcal{G}_{p,q}$ is closed under the operation of taking minors.
Thus, the edge-colored version of the graph minor theorem~(see \cite[Exercise~19.5.6]{DowneyF13})
implies that testing the membership to $\mathcal{G}_{p,q}$ is fixed-parameter tractable parameterized by $p+q$.


\section{Parameterized intractability beyond $k$}
\label{sec:para-hard}

Since {\ProblemNameShort} parameterized by~$k$ belongs to $\mathrm{W[2]}$, 
both subcases handled in this section also belong to $\mathrm{W[2]}$.
Hence, we only prove the $\mathrm{W[2]}$-hardness in each case.

For a graph $G = (V,E)$ and a set $T$ of vertex pairs, a set $X \subseteq V$ is a \emph{vertex multicut} for $T$
if no connected component of $G - X$ contains a pair in $T$.
It is known that the problem of finding a vertex multicut is fixed-parameter tractable
parameterized by its size~\cite{BousquetDT18,MarxR14}.

Observe that a vertex multiway cut for $V(T)$ is a vertex multicut for $T$, but not vice versa.
Observe also that a vertex multicut for $T$ interrupts all pairs in $T$.
Therefore, the minimum vertex \emph{multicut} size is upper-bounded by the minimum vertex \emph{multiway-cut} size,
for which {\WProblemNameShort} is fixed-parameter tractable (\cref{cor:multiway-cut}),
and lower-bounded by the solution size~$k$, for which {\ProblemNameShort} is W[2]-complete~\cite{AravindS24}.
Thus, the following result sharpens the border of (in)tractability in this parameter hierarchy.
\begin{theorem}
\label{thm:multicut}
{\ProblemName} is $\mathrm{W[2]}$-complete parameterized by
the minimum vertex multicut size of the terminal pairs~$T$.
\end{theorem}
\begin{proof}
Let $\langle U, \mathcal{F}, h \rangle$ be an instance of \textsc{Hitting Set},
where $U = \{u_{1}, \dots, u_{n}\}$ and $\mathcal{F} = \{F_{1}, \dots, F_{m}\} \subseteq 2^{U}$.
From this instance, we construct an equivalent instance $\langle G = (V,E), T, h \rangle$ of {\ProblemNameShort} as follows (see \cref{fig:hpgi}~(left)).
Let $W = \{w_{i} \mid i \in [h+3]\}$.\footnote{%
For proving \cref{thm:multicut}, it suffices to set $|W| = h+1$ instead of $h+3$.
We use $h+3$ to make the proof of \cref{thm:hpgi} simpler.}
We set $V = \mathcal{F} \cup U \cup W$.
For $u_{i} \in U$ and $F_{j} \in \mathcal{F}$, we add $\{u_{i}, F_{j}\}$ into $E$ if and only if $u_{i} \in F_{j}$.
We then add all possible edges between $U$ and $W$.
We set $T = \{\{F_{i}, w_{j}\} \mid i \in [m], \ j \in [h+3]\}$.
We can see that $W$ is a vertex multicut of $T$ with size $h+3$.

Assume that $\langle U, \mathcal{F}, h \rangle$ is a yes-instance of \textsc{Hitting Set}.
Let $S \subseteq U$ be a hitting set of $\mathcal{F}$ with $|S| \le h$.
Since $S$ is a hitting set of $\mathcal{F}$, each $F_{i} \in \mathcal{F}$ includes a member of $S$, say $u_{j}$,
and that member $u_{j}$ interrupts the pair $\{F_{i}, w_{p}\}$ for each $p \in [h+3]$.

Assume that $\langle G, T, h \rangle$ is a yes-instance of {\ProblemNameShort}\@.
We first claim that there is a solution for $\langle G, T, h \rangle$ that does not include any vertex in $\mathcal{F}$.
To see this, observe that $F_{i}$ interrupts all pairs $\{F_{i}, w_{p}\}$ for $p \in [h+3]$ but no others,
while each $u_{j} \in F_{i}$ also interrupts them.
Let $S$ be a solution for $\langle G, T, h \rangle$ with $S \cap \mathcal{F} = \emptyset$.
Since $|S| \le h < |W|$, there is a vertex $w_{p} \in W \setminus S$.
Since no other vertex $w_{q} \in W \setminus \{w_{p}\}$ belongs to the geodesic interval $I_{G}[w_{p}, F_{i}]$ for every $i \in [m]$,
the set $S \setminus W$ ($= S \cap U$) has to interrupt $\{w_{p}, F_{i}\}$ for every $i \in [m]$.
This implies that, each $F_{i}$ has a neighbor $u_{j}$ in $S \cap U$, and thus, $S \cap U$ is a hitting set of $\mathcal{F}$.
\end{proof}

\begin{figure}[tbh]
  \centering
  \includegraphics{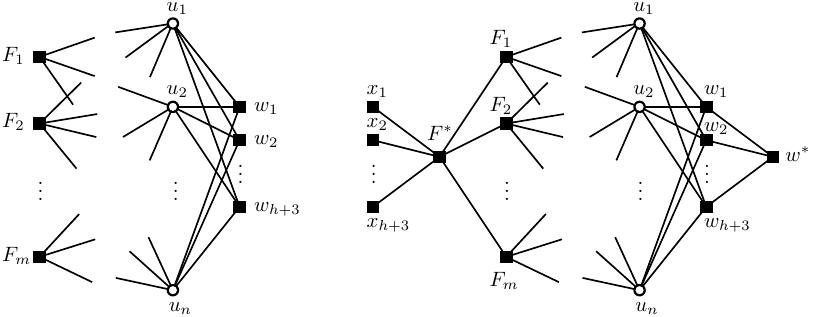}
  \caption{The graphs $G$ (left) and $H$ (right) in \cref{thm:multicut,thm:hpgi}, respectively.}
  \label{fig:hpgi}
\end{figure}

We now modify the proof of \cref{thm:multicut} and show that {\ProblemNameShort} is W[2]-complete parameterized by~$k$
even if the structure formed by the terminal pairs is quite limited, i.e., $T = \binom{Q}{2}$ for some $Q \subseteq V$.
In other words, we consider the setting where we want to hit pairwise geodesic intervals in a given vertex set $Q$.
As this variant could be of independent interest, we define it as a separate problem as follows.

\begin{tcolorbox}
\begin{description}
  \setlength{\itemsep}{0pt}
  \item[Problem:] {\PairwiseProblemName} (\PairwiseProblemNameShort)
  \item[Input:] A graph $G = (V,E)$, a set $Q \subseteq V$, and an integer $k$.
  \item[Question:] Is there $S \subseteq V$ with $|S| \le k$ such that $S \cap I_{G}[u,v] \ne \emptyset$ for each $\{u,v\} \in \binom{Q}{2}$?
\end{description}
\end{tcolorbox}

\begin{theoremrep}
\label{thm:hpgi}
{\PairwiseProblemName} is $\mathrm{W[2]}$-complete parameterized by solution size~$k$.
\end{theoremrep}
\begin{proof}
We first repeat the construction in the proof of \cref{thm:multicut}
to obtain an instance $\langle G = (V,E), T, h \rangle$ of {\ProblemNameShort}
from an instance $\langle U, \mathcal{F}, h \rangle$ of \textsc{Hitting Set}.
From $\langle G, T, h \rangle$, we construct an equivalent instance $\langle H, Q, h+2 \rangle$ of {\PairwiseProblemNameShort} (see \cref{fig:hpgi}~(right)):
to obtain $H$ from $G$, we add vertices $F^{*}$ and $w^{*}$ and a set $X = \{x_{i} \mid i \in [h+3]\}$ of vertices,
and then add all possible edges between $F^{*}$ and $X \cup \mathcal{F}$ and between $w^{*}$ and $W$;
we set $Q = V(H) \setminus U$.

Assume that $\langle G, T, h \rangle$ is a yes-instance of {\ProblemNameShort}\@.
Let $S \subseteq U$ be a solution for $\langle G, T, h \rangle$.
We set $S' = S \cup \{F^{*}, w^{*}\}$. 
We can see that $S'$ interrupts all pairs in $\binom{Q}{2}$:
$S$ interrupts all pairs in $T = \{\{F_{i}, w_{j}\} \mid i \in [m], \ j \in [h+3]\}$;
$F^{*}$ interrupts all pairs involving a vertex in $X$ and all pairs $\{F_{i}, F_{j}\}$;
and $w^{*}$ interrupts all pairs $\{w_{i}, w_{j}\}$.

Assume that $\langle H, Q, h+2 \rangle$ is a yes-instance of {\PairwiseProblemNameShort}\@.
Let $S$ be a solution for $\langle G, Q, h \rangle$.
Observe that $F^{*} \in S$ since otherwise we have to take all $h+3$ vertices of $X$ into $S$.
Similarly, we can see that $w^{*} \in S$.
Since neither $F^{*}$ nor $w^{*}$ can interrupt pairs $\{F_{i}, w_{p}\} \in T$ 
and $I_{G}[F_{i}, w_{p}] = I_{H}[F_{i}, w_{p}]$ for all $i \in [m]$ and $j \in [h+3]$,
we can conclude that $S \setminus \{F^{*}, w^{*}\}$ interrupts all pairs in $T$.
\end{proof}


\section{Concluding remarks}
\label{sec:conclusion}

In this paper, we showed that {\ProblemName} is intractable even on graphs with highly restricted structures,
while combinations of the solution size~$k$ and some structural graph parameters make the problem fixed-parameter tractable
(see \cref{fig:graph-parameters}).
As \cref{fig:graph-parameters}~(right) shows, when $k$ is part of the parameter, many cases remain unsettled.
It would be interesting to further investigate these cases.

In one of our fixed-parameter algorithms (\cref{thm:whgi-p-q-sep}),
we have a restriction that a certain kind of a separator is given as part of the input,
and then we discussed that for our target cases it is actually not a restriction (\cref{cor:vi,cor:multiway-cut}).
We wonder if requiring such a separator in the input is not a restriction in general.
More precisely, we would like to ask whether the following problem is fixed-parameter tractable parameterized by $p+q$.

\begin{tcolorbox}
\begin{description}
  \setlength{\itemsep}{0pt}
  \item[Input:] A graph $G = (V,E)$, a set of terminal vertices $R \subseteq V$, and integers $p$ and $q$.
  \item[Question:] Is there $Z \subseteq V$ with $|Z| \le p$ such that 
  each connected component $C$ of $G - Z$ contains at most $q$ terminal vertices (i.e., $|V(C) \cap R| \le q$)?
\end{description}
\end{tcolorbox}

\noindent
Note that, as we discussed in the last paragraph of \cref{sssec:app-main-alg},
the graph minor theorem gives a nonuniform fixed-parameter algorithm parameterized by $p + q$ for this problem,
while we are looking for a uniform one.

\bibliography{ref}

\end{document}